\documentclass[12pt]{amsart}

\usepackage{mathtools}

\usepackage{hyperref}
\usepackage{multimedia}
\usepackage{amsmath}
\usepackage{amsfonts}
\usepackage{amssymb}
\usepackage{amsthm}
\usepackage{latexsym}
\usepackage{color}
\usepackage{enumerate}
\usepackage{mathrsfs}
\usepackage{comment}
\usepackage{setspace}
\usepackage{appendix}
\usepackage{graphics}
\usepackage{graphicx}
 \usepackage{subfigure}
\usepackage{float}

\newtheorem{thm}{Theorem}[section]
\newtheorem{prop}{Proposition}[section]

\theoremstyle{remark}
\newtheorem{remark}{Remark}[section]

\numberwithin{equation}{section}

\def\tr{\textmd{tr}}

\def\R{\mathbb{R}}

\def\R{\mathbb{R}}

\def\vh{\vspace{.2cm}}

\def\p{\partial}

\def\m{\mathfrak{m}}

\def\Mms{M_{_{MS}}}
\def\gms{g_{_{MS}}}

\newcommand{\be}{\begin{equation}}
\newcommand{\ee}{\end{equation}}
\newcommand{\bee}{\begin{equation*}}
\newcommand{\eee}{\end{equation*}}

\def\mb{\mathfrak{m}_{_B}}

\begin{document}

\title{Bartnik mass via vacuum extensions}

\author{Pengzi Miao}
\address[Pengzi Miao]{Department of Mathematics, University of Miami, Coral Gables, FL 33146, USA}
\email{pengzim@math.miami.edu}

\author{Naqing Xie}
\address[Naqing Xie]{School of Mathematical Sciences, Fudan
University, Shanghai 200433, China}
\email{nqxie@fudan.edu.cn}

\dedicatory{Dedicated to Luen-Fai Tam on the occasion of his 70th birthday.}

\thanks{The first named author's research was partially supported by NSF grant DMS-1906423. 
The second named author's research was partially supported by National Natural Science Foundation of China \#11671089.
}


\begin{abstract}
We construct asymptotically flat, scalar flat extensions of Bartnik data $(\Sigma, \gamma, H)$,  where
$ \gamma $ is a metric of positive Gauss curvature on a two-sphere $ \Sigma$, and $H $ is a function
that is either  positive or identically zero on $\Sigma$, such that the mass of the extension can be made  
arbitrarily close to the half area radius of $(\Sigma, \gamma)$.

In the case of $ H \equiv 0 $, the result gives an analogue of a theorem of Mantoulidis and Schoen \cite{M-S},
but with extensions that have vanishing scalar curvature.
In the context of initial data sets in general relativity, the result  produces  asymptotically flat, time-symmetric, vacuum initial data with an apparent horizon $(\Sigma, \gamma)$, for any metric $\gamma$ with positive Gauss curvature, such that the mass of the initial data is arbitrarily close to the optimal value in the Riemannian Penrose inequality.

The method we use is the Shi-Tam type metric construction from \cite{ShiTam02} and a 
refined Shi-Tam monotonicity, found by the first named author in \cite{Miao09}.

\end{abstract}


\maketitle

\markboth{Pengzi Miao and Naqing Xie}{Bartnik mass via vacuum extensions}

\section{Introduction}

Let $ \Sigma$ be a two-sphere. Given a Riemannian metric $\gamma $ and a function $H $ on $\Sigma$, the Bartnik quasi-local mass \cite{Bartnik-qlmass} of the triple $(\Sigma, \gamma, H)$ can be defined as
\be
\mb(\Sigma , \gamma ,H) = \inf \left \lbrace \m (M, g) \,\vert \,(M, g) \text{ is an admissible extension of }(\Sigma, \gamma,H)\right \rbrace.
\ee
Here $ \m (\cdot)$ is the ADM mass \cite{ADM}, and
an asymptotically flat Riemannian $3$-manifold $(M, g)$ with boundary $\p M$ is  an {admissible extension of $(\Sigma, \gamma,H)$} if
\begin{itemize}
\item[i)]  $g$  has nonnegative scalar curvature;
\item[ii)]   $\p M$ with the induced metric is isometric to  $(\Sigma, \gamma)$ and, under the isometry, the mean curvature of $\p M$ in
$(M, g $) equals $H$; and
\item[ iii)]  $(M, g)$ satisfies certain non-degeneracy condition that prevents $\m (M, g)$ from being arbitrarily small; for instance,
it is often required that $(M, g)$  contains no closed minimal surfaces (enclosing $\p M$), or $\p M$ is outer-minimizing in $(M, g)$.
\end{itemize}
We refer readers to  \cite{Anderson-Jauregui, Jauregui-18, McCormick-18} for discussion on the numerous variations in the definition of Bartnik mass, and the recent progress on reconciling them.

Let $ K_\gamma $ denote the Gauss curvature of $(\Sigma, \gamma)$. If $ K_\gamma > 0 $, by
the work of Nirenberg \cite{Nirenberg} (and also  Pogorelov \cite{Pogorelov}), $(\Sigma, \gamma)$
admits an isometric embedding into $ \R^3$ as a convex surface, unique up to  rigid motions.
Let $ H_0 $ be the mean curvature of such an isometric embedding. In \cite{ShiTam02}, 
as a key step in their proof of the positivity of the Brown-York mass \cite{BY1, BY2}, 
Shi and Tam proved the following:

\begin{thm}[Shi-Tam \cite{ShiTam02}] \label{thm-S-T}
Suppose $K_\gamma > 0 $. Identify  $ \Sigma$ with  the image of
the isometric embedding of $(\Sigma, \gamma)$ in $ \R^3$ and
write the Euclidean metric $g_{_E} = d r^2 + g_r $ on  $E$ that  is the exterior of $\Sigma$.
Here $ g_r$ is the induced metric on the surface $ \Sigma_r $ that is $r$-distance away from $ \Sigma$ in $(E, g_{_E})$.
Given any function $ H > 0 $ on $ \Sigma$, there is a unique function $ u > 0 $ on $ E $ such that the metric $ g_u = u^2 d \rho^2 + g_\rho $
is asymptotically flat  and
\begin{itemize}
\item $g_u$ has zero scalar curvature;
\item
the mean curvature of $ \Sigma = \p E $ in $(E, g_u)$ equals $ H$; and
\item the ADM mass $ \m (E, g_u)$ satisfies
$$
 \m (E, g_u) \le \frac{1}{8 \pi} \int_{\Sigma} ( H_0 - H) \, d \sigma ,
$$
where $ d \sigma$ is the area measure on $(\Sigma, \gamma)$.
\end{itemize}
Consequently, since  $(E, g_u)$   is foliated by a $1$-parameter family of surfaces $\{ \Sigma_r \}_{r \ge 0}$ with positive mean curvature,
it follows that
\be \label{eq-mb-mby}
 \mb(\Sigma, \gamma, H) \le \frac{1}{8 \pi} \int_\Sigma ( H_0 - H) \, d \sigma
\ee
 regardless of the non-degeneracy condition iii) used in the definition of $\mb(\Sigma, \gamma, H)$.
\end{thm}

In this paper, we apply the method of Shi-Tam in \cite{ShiTam02} and its variation by the first named author in \cite{Miao09}
to exhibit suitable  extensions of $(\Sigma, \gamma, H)$ such that the ADM mass of these extensions
is controlled by $ |\Sigma|_\gamma$, the area of $\gamma$.

\begin{thm} \label{thm-main}
Let $ \gamma $ be a metric with $ K_\gamma >  0 $ and $ H$ be a positive function on $ \Sigma$.
Given any $ \epsilon > 0 $, there exists an asymptotically flat manifold $(M, g) $, diffeomorphic to $ \Sigma \times [1, \infty)$, 
such that
\begin{itemize}
\item[i)]  $ g $  has zero  scalar curvature;
\item[ii)]   $\p M$ with the induced metric is isometric to  $(\Sigma, \gamma)$ and the mean curvature of $\p M$ in
$(M, g $) equals $H$ under the isometry;
\item[ iii)]  $(M, g)$  is foliated by a $1$-parameter family of closed surfaces
with positive mean curvature; and
\item[iv)]
the mass of  $ (M, g)$ satisfies
\be \label{eq-mass-M-S-main}
\m (M, g) \le \sqrt{ \frac{ | \Sigma |_{\gamma} }{16 \pi } } + \epsilon .
\ee
\end{itemize}
Consequently, the Bartnik mass of $(\Sigma, \gamma, H)$ satisfies
\be \label{eq-mb-half-r}
 \mb(\Sigma, \gamma, H)  \le \sqrt{ \frac{ | \Sigma |_{\gamma} }{16 \pi } } .
\ee
\end{thm}

Theorem \ref{thm-main} also holds if $H$ is identically zero. 
We state this case separately 
since the extensions  in this setting represent  time-symmetric,
vacuum initial data with apparent horizon boundary.

\begin{thm} \label{thm-main-H-zero}
Let $ \gamma $ be a metric with $ K_\gamma >  0 $ on $\Sigma$.
Given any $ \epsilon > 0 $, there exists an asymptotically flat manifold $(M, g) $, diffeomorphic to $ \Sigma \times [1, \infty)$, 
such that
\begin{itemize}
\item[i)]  $ g $  has zero  scalar curvature;
\item[ii)]   $\p M$ with the induced metric is isometric to  $(\Sigma, \gamma)$ and 
$\p M$ has zero mean curvature in $(M, g)$; 
\item[ iii)]  the interior of $(M, g)$  is foliated by a $1$-parameter family of closed surfaces
with positive mean curvature; and
\item[iv)]
the mass of  $ (M, g)$ satisfies
\be \label{eq-mass-M-S-main-H-zero}
\m (M, g) \le \sqrt{ \frac{ | \Sigma |_{\gamma} }{16 \pi } } + \epsilon .
\ee
\end{itemize}
Consequently, the Bartnik mass of $(\Sigma, \gamma, 0)$ satisfies
\be \label{eq-mb-half-r-H-zero}
 \mb(\Sigma, \gamma, 0)  \le \sqrt{ \frac{ | \Sigma |_{\gamma} }{16 \pi } } .
\ee
\end{thm}

We give two important remarks regarding Theorems \ref{thm-main} and \ref{thm-main-H-zero}.

\begin{remark}
In \cite{M-S}, Mantoulidis and Schoen proved the following result:
if  $\gamma$ is a metric on $ \Sigma$ such that
$ \lambda_1 ( -\Delta_\gamma + K_\gamma) > 0 $, where $\Delta_\gamma$ is the Laplacian of $(\Sigma, \gamma)$,
then given any $\epsilon > 0 $, there exists an asymptotically flat manifold $(M_{_{MS} }, g_{_{MS}}  )$ such that
\begin{itemize}
\item[a)]  $ \gms $  has nonnegative   scalar curvature, and has strictly positive scalar curvature somewhere;
\item[b)]   $\p \Mms$ with the induced metric is isometric to  $(\Sigma, \gamma)$ and the mean curvature function of $\p \Mms$ in
$(\Mms, \gms $) equals $0$; 
\item[c)]  $(\Mms, \gms)$  is foliated by a $1$-parameter family of closed surfaces
with positive mean curvature, and $(\Mms, \gms)$, outside a compact set, is isometric to a spatial Schwarzschild manifold
with mass $m$; and 
\item[d)]   $  \m( \Mms, \gms) = m $ satisfies
\be \label{eq-m-M-S}
m  <  \sqrt{ \frac{ | \Sigma |_{\gamma} }{16 \pi }  } + \epsilon .
\ee
\end{itemize}
In the context of the Bartnik mass, Mantoulidis-Schoen's result shows
\be \label{eq-mb-MS-bound}
\mb (\Sigma, \gamma, 0) \le  \sqrt{ \frac{ | \Sigma |_{\gamma} }{16 \pi }  }
\ee
under the assumption $ \lambda_1 ( -\Delta_\gamma + K_\gamma ) > 0 $.

Comparing Theorem \ref{thm-main-H-zero} with Mantoulidis-Schoen's theorem, one sees 
Theorem \ref{thm-main-H-zero} proves \eqref{eq-mb-MS-bound} under
a much stronger assumption $ K_\gamma > 0$. However, the extension $(\Mms, \gms)$ in \cite{M-S}
has positive scalar curvature somewhere, while $(M, g)$ in Theorem \ref{thm-main-H-zero} has identically zero scalar curvature.

In the context of initial data sets in general relativity,  Theorem \ref{thm-main-H-zero} produces  asymptotically flat, time-symmetric,
vacuum  initial data with an apparent horizon $(\Sigma, \gamma)$ whenever $ K_\gamma > 0$ such that the mass of the initial data is arbitrarily close to the optimal value in the Riemannian Penrose inequality \cite{Bray01, H-I01}.
\end{remark}

\begin{remark}
In light of \cite{M-S}, the  bound \eqref{eq-mb-half-r} on $\mb (\Sigma, \gamma, H)$ for a positive $H$
is well expected and is known among experts (see \cite{Jauregui-18, McCormick-18} for instance).
This is because  the manifold  $(\Mms, \gms)$ from \cite{M-S} is ``almost" an admissible extension of
$(\Sigma, \gamma, H > 0)$, except  the mean curvature of $\p \Mms$ in
$(\Mms, \gms)$ is zero. 
Thus, if one enlarges the class of admissible extensions in the definition of $\mb (\Sigma, \gamma, H)$ by replacing condition ii)
with a condition
\begin{itemize}
\item[$\widetilde{\text{ii}})$]   $\p M$ with the induced metric is isometric to  $(\Sigma, \gamma)$ and  $ H \ge H_{\p M}$
\end{itemize}
and denotes the resulting Bartnik mass by $ \tilde{\m}_{_B}  (\Sigma, \gamma, H) $, then
$(\Mms, \gms)$ would be a legitimate admissible extension of $(\Sigma, \gamma, H)$ and hence, by \eqref{eq-m-M-S},  one will have
\be \label{eq-tilde-mb-half-r}
\tilde{\m}_{_B} (\Sigma, \gamma, H) \le  \sqrt{ \frac{ | \Sigma |_{\gamma} }{16 \pi }  }.
\ee
The geometric meaning of  ``$ H \ge H_{\p M}$" used in condition $\widetilde{\text{ii}})$ can be found in \cite{Miao02}.
Naturally one would like to know if $ \tilde{\m}_{_B} (\Sigma, \gamma, H) $ agrees with   $\mb (\Sigma, \gamma, H)$.
We refer readers to the recent work of Jauregui \cite{Jauregui-18} and McCormick \cite{McCormick-18}
for results pertinent to this question.

Even if $\tilde{\m}_{_B} (\Sigma, \gamma, H) = \mb (\Sigma, \gamma, H)$, we note Theorem \ref{thm-main}  reveals
more information  than \eqref{eq-mb-half-r}. This is again because the extension $(M, g)$ in Theorem \ref{thm-main} has zero scalar curvature.
Suppose  $ K_\gamma > 0 $, if one shrinks the class of
admissible extensions in the definition of $\mb (\Sigma, \gamma, H)$ by replacing condition i)
with a condition
\begin{itemize}
\item[$\widetilde{\text{i}})$]   $g$ has zero scalar curvature
\end{itemize}
and denotes the resulting Bartnik mass by $ {\m}^0_{_B}  (\Sigma, \gamma, H) $, then
 Theorem \ref{thm-main} still applies to show
 \be \label{eq-0-mb-half-r}
  {\m}^0_{_B}  (\Sigma, \gamma, H) \le  \sqrt{ \frac{ | \Sigma |_{\gamma} }{16 \pi }  } .
 \ee
On the other hand, as the metric $(\Mms , \gms) $  has strictly positive scalar curvature somewhere,
it is not clear if the result  in \cite{M-S}  could imply  \eqref{eq-0-mb-half-r}.

In the definition of $\mb(\cdot)$, it is indeed natural to restrict the class of admissible extensions of $(\Sigma, \gamma, H)$
to extensions with identically zero scalar curvature. It was conjectured by Bartnik \cite{Bartnik-qlmass} that,
under suitable assumptions on $(\gamma, H)$, a minimizer $(M, g)$ achieving  $\mb(\Sigma, \gamma, H)$
exists and is a static vacuum initial data set.
For this reason, it is reasonable to consider the revised variational problem of minimizing mass 
over extensions with zero scalar curvature. 
Besides Theorem \ref{thm-S-T} of Shi-Tam, prior results on estimating the Bartnik mass by constructing 
scalar flat extensions using PDE methods were also given by Lin and Sormani \cite{L-S14}. 
\end{remark}

We end this section by comparing estimate  \eqref{eq-mb-half-r} and \eqref{eq-mb-mby}.
By the classic Minkowski inequality, if $\Sigma$ is a closed convex surface in $ \R^3$ 
with intrinsic metric $\gamma$ and mean curvature $H_0$, then
$$
\frac{1}{8\pi} \int_\Sigma H_0 \ d \sigma \ge \sqrt{ \frac{ | \Sigma |_\gamma} { 4 \pi } } .
$$
Thus, if $H > 0$ is relative small, i.e.
if  $ \displaystyle \frac{1}{8\pi} \int_\Sigma H \ d \sigma <   \sqrt{ \frac{ | \Sigma |_\gamma} { 16 \pi } } $, then
\be
\frac{1}{8 \pi} \int_\Sigma ( H_0 - H) \, d \sigma > \sqrt{ \frac{ | \Sigma |_\gamma} { 16 \pi } } .
\ee
In this case,  \eqref{eq-mb-half-r} is an estimate  sharper than  \eqref{eq-mb-mby}.
On the other hand, if $H $ is close to $H_0$,  then    \eqref{eq-mb-mby} represents  a better estimate.

\section{Proof of Theorems \ref{thm-main} and \ref{thm-main-H-zero}}

For simplicity, replacing $(\gamma, H)$ with $(c^2 \gamma, c^{-1} H )$ for a constant $c > 0 $, we may  assume
 $ | \Sigma |_\gamma = 4 \pi $.

We first prove Theorem \ref{thm-main} in which $ H $ is a positive function on $\Sigma$.
Suppose $K_\gamma > 0 $.  Let  $\{ g(t) \}_{0 \le t \le 1} $ be a fixed smooth path of metrics
 on $ \Sigma$ satisfying
\begin{enumerate}
\item[(i)]  $g(0) = \gamma $, $g(1) $ is a round metric;
\item[(ii)] $K_{ g(t) } > 0 $, $ \forall \  t \in [0,1]$; and
\item[(iii)] $ \tr_{g(t)} g'(t) = 0 $, $ \forall \ t \in [0,1]$.
\end{enumerate}
(Existence of such a path can be given by the proof of \cite[Lemma 1.2]{M-S} for instance.)
It follows from conditions (i) and (iii) that $g(1) = \sigma$, a standard round metric on $ \Sigma$ with area $4\pi$.

Given any $\delta > 0 $, let
$$ t_\delta  = t_\delta ( \rho) : [1, \infty) \rightarrow [0, 1] $$
be a smooth function
such that
\be \label{eq-t-delta-rho-1}
t_\delta (1) = 0, \ \ \mathrm{and} \ \ t_\delta ( \rho ) = 1, \ \forall \ \rho \ge 1 + \delta.
\ee
Define
\be \label{eq-g-delta-rho}
g_{\delta, \rho} = g (  t_\delta ( \rho ) ) , \ \rho \in [1, \infty) .
\ee

In what follows, we suppress the notation $\delta$ and denote  $ g_{\delta, \rho}$ by $ g_\rho $.
$ \{ g_\rho \}_{1 \le \rho < \infty} $ is a smooth family of metrics on $\Sigma$ satisfying
\be
g_1 = \gamma, \ \mathrm{and} \
g_\rho = \sigma, \ \forall \ \rho \ge 1 + \delta .
\ee
On $ N =  [1, \infty) \times \Sigma $, consider a background metric
\be
g^b = d \rho^2 + \rho^2 g_\rho .
\ee
Let  $N_{\delta} = [1 + \delta, \infty) \times \Sigma $, then
\be
g^b = d \rho^2 + \rho^2 \sigma = g_{_E}, \ \mathrm{on} \ N_\delta,
\ee
where $ g_{_E}$ is the standard Euclidean metric.

For each $\rho$, let $H^b$ and $ A^b$  denote the mean curvature and the second fundamental form
of $ \Sigma_\rho = \{ \rho \} \times \Sigma$
in $(N, g^b)$ with respect to $ \frac{\p}{\p \rho}$, respectively. Then
\be
H^b = 2 \rho^{-1} + \frac12 \tr_{g_\rho} (\p_\rho g_\rho ) = 2 \rho^{-1},
\ee
where we used (iii).
Given any function $ u = u (\rho, x) > 0 $ on $N$, following Bartnik \cite{Bartnik93} and Shi-Tam \cite{ShiTam02},
we consider  the metric
\be
g_u = u^2 d \rho^2 + \rho^2 g_\rho .
\ee
The induced metric from $g_u $ on $ \Sigma_1 $, which is identified with $ \Sigma$, is $g_1  = \gamma $.
If $ H_u $ denotes the mean curvature of $ \Sigma_\rho $ in $(N, g_u)$,
then
\be
 H_u = u^{-1} H^b > 0 .
\ee

The following claim follows directly from results in \cite{ShiTam02} and \cite{EMW09}.

\begin{prop} \label{prop-u-existence}
Given the pair $(\gamma, H)$ on $ \Sigma$,
there exists a function $ u  > 0 $ on $N$ such that
\begin{itemize}
\item $ g_u = u^2 d \rho^2 + \rho^2 g_\rho $  has zero scalar curvature, and
\item $ H_u = H $ at $ \Sigma = \p N$.
\end{itemize}
Moreover, $ u \to 1 $ as $ \rho \to \infty$ and $ (N, g_u)$ is asymptotically flat and is foliated
by $\{ \Sigma_\rho \}_{1 \le \rho < \infty}$ with positive mean curvature.
\end{prop}

\begin{proof}
We recall the PDE that $u$ needs to satisfy so that $ g_u$ has zero scalar curvature.
Let $ \tilde g_\rho = \rho^2 g_\rho $.  By \cite[Equation (1.10)]{ShiTam02} (also see \cite[Equation (5)]{EMW09}),
\be \label{eq-EMW}
H^b \p_\rho  u = u^2 \Delta_{\tilde g_\rho}  u - K_{\tilde g_\rho}  u^3 + \frac{1}{2} \left[ (H^b)^2 + | A^b |^2_{\tilde g_\rho} + 2
\p_\rho H^b \right] u .
\ee
Since $ H^b > 0 $ and $ K_{\tilde g_\rho}  = \rho^{-2} K_{g_\rho} > 0 $ for every $\rho$,
 a positive solution  $u$ with an initial condition $ u |_{\Sigma} = H^{-1} H^b $ exists on $[1, T] \times \Sigma$ for all $ T > 0$
 by \cite[Proposition 2]{EMW09}.
Since  $ g^b $ is the Euclidean metric on $N_{\delta}$,
the claim on the asymptotic behavior of $u$ follows from \cite[Theorem 2.1]{ShiTam02}.
\end{proof}

Let $ u $ be given in  Proposition \ref{prop-u-existence},  $(N, g_u)$
is an admissible extension of $(\Sigma, \gamma, H)$. By definition,
\be \label{eq-B-mass-by-gu}
\mb (\Sigma, \gamma, H) \le \m ( g_u) .
\ee

We want to estimate $\m (g_u) $ in terms of data at $ \Sigma_{1 + \delta} = \p N_{ \delta} $.
Initially, we could apply \cite[Lemma 4.2]{ShiTam02} to have
\be
 \int_{\Sigma_\rho} (H^b - H_u) \, d \sigma_\rho \ \mathrm{is \ monotone \ nonincreasing \ in \ }
 \rho \in [1+\delta, \infty) ,
\ee
where $ d \sigma_\rho$ is the area measure on $(\Sigma_\rho, \rho^2 g_\rho)$.
Also, by \cite[Theorem 2.1]{ShiTam02},
\be
\lim_{\rho \to \infty} \frac{1}{8\pi} \int_{\Sigma_\rho} (H^b - H_u) \, d \sigma_\rho  = \m (g_u) .
\ee
Thus,
\be \label{eq-mass-gu-est-st}
\begin{split}
\m (g_u) \le & \ \frac{1}{8\pi} \int_{\Sigma_{1+ \delta}} (H^b - H_u)  \, d \sigma_{1+\delta} \\
= & \ (1+\delta) -  \frac{1}{8\pi} \int_{\Sigma_{1+ \delta}} H_u \, d \sigma_{1+\delta} .
\end{split}
\ee
Since $ H_u > 0 $, omitting the term involving $ H_u$, one has
$$ \m (g_u)  < 1 + \delta  .$$
However, this estimate   is not sufficient to show \eqref{eq-mass-M-S-main}.

To verify \eqref{eq-mass-M-S-main}, we make use of a refined monotonicity property concerning $(N_\delta, g_u)$.
The next proposition is essentially a restatement of  results from \cite[Proposition 1 and Equation (56)]{Miao09}.

\begin{prop} [\cite{Miao09}]  \label{prop-LRPI}
Given any constant $ r > 0 $,  let $E_r = \R^3 \setminus \{ | x | < r \}$.
Let
$ g^b_{0} = d \rho^2 + \rho^2 \sigma $ denote a background  Euclidean metric on $E_r$.
On $E_r$, let $ u > 0 $ be a function such that
the metric
$ g_u = u^2 d \rho^2 + \rho^2 \sigma $
has zero scalar curvature and hence is asymptotically flat.
Let $ H_u = 2 \rho^{-1} u^{-1}  $ be the mean curvature
of $ S_\rho = \{ | x | = \rho \}$ in $(E_r, g_u)$.
Then, for any constant $ m \in (-\infty, \frac12 r]$,
\be \label{eq-LRPI-est-1}
\m(g_u) \le m + \frac{1}{8\pi} \int_{S_r} (2 r^{-1} N - H_u) N  \, d \sigma_r
\ee
where $ N = \sqrt{ 1 - \frac{2m}{r} }$.
As a result, by minimizing its right side over $ m $, one has
\be \label{eq-mass-gu-refined}
\m (g_u) \le \sqrt{ \frac{|S_r|}{16 \pi} } \left[ 1 - \frac{1}{16 \pi | S_r| } \left( \int_{S_r} H_u \, d \sigma_r \right)^2 \right] .
\ee
\end{prop}

\begin{proof}
Given any $ m \in ( - \infty, \frac12 r] $,  consider a  background
Schwarzschild metric
\be
g^b_m = \frac{1}{ 1 - \frac{2m}{\rho} }  d \rho^2 + \rho^2 \sigma
\ee
on $E_r$.
 Let $ H^b_m $ be the mean curvature
of $S_\rho $ in $(E_r, g^b_m )$. By \cite[Proposition 1]{Miao09},
\be \label{eq-LRPI-est-2}
\int_{ S_\rho }  (H^b_m - H_u) N(\rho) \, d \sigma_\rho \ \mathrm{is \ monotone \ nonincreasing \ in \ }
\rho \in [r , \infty),
\ee
where $N  (\rho) =  \sqrt{1 - \frac{2m}{\rho} } $, and
\be \label{eq-LRPI-est-3}
\lim_{\rho \to \infty} \frac{1}{8 \pi} \int_{S_\rho}  (H^b_m - H_u) N(\rho) \, d \sigma_\rho = \m( g_u) - m .
\ee
Thus, by \eqref{eq-LRPI-est-2}, \eqref{eq-LRPI-est-3} and the fact $ H^b_m = 2 \rho^{-1} N(\rho)$,
\be \label{eq-LRPI-est-4}
\begin{split}
\m (g_u) \le & \  m + \frac{1}{8 \pi} \int_{S_r}  (H^b_m - H_u) N \, d \sigma_r  \\
= & \ m + \frac{1}{8 \pi} \int_{S_r}  ( 2 r^{-1} N - H_u) N \, d \sigma_r  .
\end{split}
\ee
Denote the right side of \eqref{eq-LRPI-est-4} by $\Phi(m)$, i.e.
\be
\Phi (m) =  m  + r N^2 -  \left( \frac{1}{8 \pi} \int_{S_r}   H_u \, d \sigma_r  \right) N .
\ee
Since $ N^2 = 1 - \frac{2m}{r}$, one can rewrite $\Phi(m)$ as
\be
\Phi (m)
=  \frac12 r \left[ N - \frac{1}{8 \pi r} \left( \int_{S_r} H_u \, d \sigma_r \right) \right]^2
+ \frac{1}{2} r \left[ 1 -   \left( \frac{1}{8\pi r } \int_{S_r} H_u \, d \sigma_r \right)^2 \right] .
\ee
Minimizing $\Phi (m)$ over $ m \in ( - \infty, \frac12 r]$,  one has
\be
\min_{m} \Phi (m) = \frac{1}{2} r \left[ 1 -   \left( \frac{1}{8\pi r } \int_{S_r} H_u \, d \sigma_r \right)^2 \right] .
\ee
This explains \eqref{eq-mass-gu-refined}.
\end{proof}

\begin{remark}
A monotonicity generalizing \eqref{eq-LRPI-est-2} with the background Schwarzschild metric replaced by 
a general static metric can be found in \cite{Lu-Miao}.
\end{remark}

We resume to prove  \eqref{eq-mass-M-S-main}. Applying Proposition \ref{prop-LRPI}
to $(N_{\delta}, g_u)$, we have
\be \label{eq-mass-gu-est-2}
\m (g_u)  \le
\sqrt{ \frac{|\Sigma_{1+\delta} | }{16 \pi} }
\left[ 1 - \frac{1}{16 \pi | \Sigma_{1 + \delta } | } \left( \int_{\Sigma_{1+\delta}} H_u \, d \sigma_{1+\delta}  \right)^2 \right] .
\ee
Omitting the term involving $ H_u$, we conclude from \eqref{eq-mass-gu-est-2} that
\be \label{eq-b-mass-3}
\m (g_u) < \sqrt{ \frac{|\Sigma_{1+\delta} | }{16 \pi} }  = \frac12 (1 + \delta) .
\ee
Since $ \delta$ can be arbitrarily small, \eqref{eq-mass-M-S-main}  follows from \eqref{eq-b-mass-3} and rescaling.
Theorem \ref{thm-main} is proved.

\medskip

Next,  we prove Theorem \ref{thm-main-H-zero} in which $H \equiv 0 $. In this case, to obtain a solution $ u$ 
satisfying Proposition \ref{prop-u-existence}, one needs
to solve \eqref{eq-EMW} with an initial condition $ u |_\Sigma = \infty$. While this might be achieved by following the proofs
in \cite{Bartnik93, ShiTam02}, we revise our choice of $ \{g_\rho\}$ to
apply a result of Smith \cite{Smith09, Smith11} which imitates  the horizon in a spatial Schwarzschild manifold.
More precisely,  given any small $\delta > 0 $, we modify the choice of $t_\delta (\rho) $ in \eqref{eq-t-delta-rho-1}
by re-defining it so that
\be \label{eq-t-delta-rho-2}
\left\{
\begin{array}{ll}
t_\delta (\rho) = 0,  & \  \text{if} \ 1 \le \rho \le 1 + \frac{1}{2}\delta, \\
t_\delta (\rho) = 1, & \   \text{if} \ \rho \ge 1 + \delta .
\end{array}
\right.
\ee
With this  choice of $t_\delta (\rho)$, define $ g_{\rho} = g ( t_\delta (\rho) ) $ as in \eqref{eq-g-delta-rho}.
Then $\{ g_\rho \}_{\rho \ge 1}$ satisfies
\be \label{eq-g-delta-rho-2}
\left\{
\begin{array}{ll}
g_\rho = \gamma,  & \  \text{if} \ 1 \le \rho \le 1 + \frac{1}{2}\delta, \\
g_\rho = \sigma , & \   \text{if} \ \rho \ge 1 + \delta .
\end{array}
\right.
\ee
On $ N =  [1, \infty) \times \Sigma $, consider the  background metric
$ g^b = d \rho^2 + \rho^2 g_\rho $.  Applying \cite[Main Theorem]{Smith09},
 there exists a function $ u$  of the form
\be
u = \frac{v}{ \sqrt{ 1 - \frac{1}{\rho} } },
\ee
where $ v > 0$ is a smooth function on $N$, independent on $ \rho $ in $[1, 1 + \frac{1}{2} \delta]$,
such that the metric
$$ g_u = u^2 d \rho^2 + \rho^2 g_\rho $$
has zero scalar curvature. This $(N, g_u)$ satisfies Proposition \ref{prop-u-existence} with $H \equiv 0 $.
The rest of the proof is then identical to that following Proposition \ref{prop-LRPI} above.
This proves Theorem \ref{thm-main-H-zero}.

\begin{remark}
In Theorems \ref{thm-main} and \ref{thm-main-H-zero}, the assumption $ K_\gamma > 0$  is imposed to have $ K_{g(t)} > 0$, which is  to
guarantee the existence of the solution $u$ to \eqref{eq-EMW}.
\end{remark}

\begin{remark}
By  (iii) and  \eqref{eq-g-delta-rho}, equation \eqref{eq-EMW} on $ [1, \infty) \times \Sigma$  takes an explicit form of 
\be \label{eq-EMW-2}
2 \rho \, \p_\rho  u = u^2  \Delta_{g_\rho}  u -  K_{g_\rho}  u^3 + 
 \left[ 1 + \frac18 \rho^2 \left | \frac{d g}{dt} \right |_g^2 \left( \frac{d t_\delta }{d \rho} \right)^2   \right] u .
\ee
On $[1 + \delta, \infty) \times \Sigma$, it reduces to
\be \label{eq-EMW-3}
2 \rho \, \p_\rho  u = u^2  \Delta_{\sigma}  u  +   ( u - u^3)  ,
\ee
which is the equation in \cite[Example 1]{ShiTam02}.
\end{remark}

\vh

\section*{Acknowledgements}
Dr. Luen-Fai Tam is an inspiration  and a mentor to both of us. It is our great pleasure to dedicate this paper
to him on the occasion of his 70th birthday.

\vh

We also thank Dr. Christos Mantoulidis for his helpful comments on a preliminary version of this paper.

\end{document}